\documentclass{my-elsarticle}

\usepackage{amsmath,amsthm,amssymb}
\usepackage{paralist}
\usepackage{graphics} 
\usepackage{epsfig} 
\usepackage{graphicx}  
\usepackage{epstopdf}
\usepackage{a4wide}
\usepackage[colorlinks=true]{hyperref}
\hypersetup{urlcolor=blue, citecolor=red}

\allowdisplaybreaks

\newtheorem{theorem}{Theorem}[section]

\newtheorem{lemma}[theorem]{Lemma}

\theoremstyle{definition}

\newtheorem{remark}[theorem]{Remark}

\begin{document}
	
\begin{frontmatter}
		
\title{A New Compartmental Epidemiological Model for COVID-19\\ with a Case Study of Portugal}

\author[Add:a]{\href{https://orcid.org/0000-0003-4592-3426}{Ana P. Lemos-Pai\~{a}o}}
\ead{anapaiao@ua.pt}

\author[Add:a]{\href{https://orcid.org/0000-0002-7238-546X}{Cristiana J. Silva}}
\ead{cjoaosilva@ua.pt}

\author[Add:a]{\href{https://orcid.org/0000-0001-8641-2505}{Delfim F. M. Torres}\corref{corD}}
\ead{delfim@ua.pt}
\cortext[corD]{Corresponding author: delfim@ua.pt}

\address[Add:a]{Center for Research and Development in Mathematics and Applications (CIDMA),\\ 
Department of Mathematics, University of Aveiro, 3810-193 Aveiro, Portugal}


\begin{abstract}
We propose a compartmental mathematical model for the spread 
of the COVID-19 disease, showing its usefulness with respect
to the pandemic in Portugal, from the first recorded case 
in the country till the end of the three states of emergency.
New results include the compartmental model, described 
by a system of seven ordinary differential equations;
proof of positivity and boundedness of solutions;
investigation of equilibrium points and their stability analysis;
computation of the basic reproduction number; 
and numerical simulations with official real data
from the Portuguese health authorities.
Besides completely new, the proposed model 
allows to describe quite well the spread of COVID-19 in Portugal, 
fitting simultaneously not only the number of active infected individuals 
but also the number of hospitalized individuals, respectively with a $L^2$ 
error of $9.2152e-04$ and $1.6136e-04$ with respect to the initial population.
Such results are very important, from a practical point of view, 
and far from trivial from a mathematical perspective. Moreover, the obtained 
value for the basic reproduction number is in agreement with the one 
given by the Portuguese authorities at the end of the emergency states.
\end{abstract}



\end{frontmatter}


\section{Introduction}

The COVID-19 pandemic, also known as the coronavirus pandemic, 
is an ongoing pandemic of coronavirus disease 2019 (COVID-19) 
caused by severe acute respiratory syndrome coronavirus 2 (SARS-CoV-2). 
The outbreak was identified in Wuhan, China, in December 2019.
The World Health Organization (WHO) declared the outbreak a Public Health Emergency 
of international concern on 30 January 2020, and a pandemic on 11 March 2020. 

It is unanimous that the COVID-19 pandemic is the most significant global crisis
of all human history, exceeding the size and range of the repercussions 
of a World War: it has literally affected all the countries of our planet
with serious health, social, and economic consequences. 
On June 4 2020, 213 of the 247 countries or territories recognized by the United Nations 
had at least one officially case of infection by COVID-19 and 185 had 
recorded at least one death victim. Just focusing on the health consequences 
of the pandemic, it is consensual they are devastating: COVID-19 deaths 
exceeded 350,000 by June 2020 and continue to grow \cite{Worldometer}.

If the COVID-19 pandemic has changed society worldwide, 
mathematics is no exception. Changes are not only at the level
of universities and research centers physically closed, of international 
traveling ban or on the rise of online seminars and conferences. 
It is also apparent that COVID-19 has called attention of society about
the importance of mathematics in epidemiology and the relevance 
of mathematical modeling for a better understanding of the COVID-19 
health crisis \cite{CEM}. 

The spread of infectious diseases has been studied 
by mathematicians for a long time. The earliest account 
of mathematical modeling of the spread of a disease 
was carried out in 1766 by Daniel Bernoulli: he 
created a mathematical model to defend the practice 
of inoculating against smallpox, showing that universal 
inoculation would increase the life expectancy 
from 26 to 29 years \cite{Hethcote}. 

The early 20th century saw the emergence of mathematical compartmental models. 
The Kermack--McKendrick epidemic model (1927) and the Reed--Frost epidemic model (1928), 
both describe the relationship between susceptible, infected and immune individuals 
in a population. In particular, the Kermack--McKendrick epidemic model
has shown to be successful in predicting the behavior of many observed 
recorded epidemics \cite{MR1822695}. These basic but fundamental
Kermack--McKendrick type models still serve the basis to current 
mathematical research on epidemiology, e.g., for Ebola \cite{MR3810766}, 
TB \cite{MR3999702}, HIV \cite{MR3813045}, and cholera \cite{MR3602689}.

The use of quarantine for controlling epidemic diseases 
has always been controversial, because such strategy raises political, 
ethical, and socio-economic issues and requires a careful balance between 
public interest and individual rights \cite{Tognotti:quarantine}. 
Quarantine is adopted as a mean of separating persons, animals, 
and goods that may have been exposed to a contagious disease. Since 
the fourteenth century, quarantine has been the cornerstone 
of a coordinated disease-control strategy, including isolation, 
sanitary cordons, bills of health issued to ships, fumigation, 
disinfection, and regulation of groups of persons who were believed 
to be responsible for spreading of the infection \cite{Matovinovic,Tognotti:quarantine}.
Never before quarantine was adopted so widely as with COVID-19:   
more than half of the entire mankind has been affected by drastic restrictions 
in their movements and social relationships because of COVID-19 \cite{BOCCALETTI2020109794}.

Recently, several compartmental models have been proposed for COVID-19 spread 
in different regions of the world: see, e.g., \cite{Contreras} 
for a multi-group SEIRA model; \cite{Maier} for the impact of effective 
containment in the growth of confirmed cases in the outbreak in China; 
\cite{Crokidakis} for a COVID-19 spread study in Brazil; \cite{R3} 
for a COVID-19 model that considers media coverage effects; 
\cite{MR4160225:R2} for a within-host model, which describes 
the interactions between SARS-CoV-2, host pulmonary epithelial cells, 
and cytotoxic T lymphocyte cells; and \cite{mor:r3}, where 
a stochastic time-delayed model for the effectiveness of Moroccan 
COVID-19 deconfinement strategy is proposed and investigated.

Here we use mathematical modeling and compartmental models,
through quarantine, to study the COVID-19 pandemic in Portugal.
On 2 March 2020, the SARS-CoV-2 virus was confirmed 
to have reached Portugal, when it was reported that two men, 
a 60 year-old doctor, who traveled to the north of Italy on vacation, 
and a 33 year-old man working in Spain, tested positive for COVID-19.
March 12, the Portuguese government declared the highest level of alert 
and said it would be maintained until 9 April, 
because community transmission was detected.
On March 18, 2020, the President of the Portuguese Republic, Marcelo Rebelo de Sousa, 
declared all Portuguese territory in a State of Emergency 
for the following fifteen days, with the possibility of renewal. 
This was the first State of Emergency since the Carnation Revolution in 1974.
By March 24, the Portuguese Government admitted that the country 
could not contain the virus any longer,
and on March 26 the country entered the ``Mitigation Stage''. 
April 2 the Parliament approved the extension of the State of Emergency, 
as requested by the President. This State of Emergency remained until 17 April, 
being then renovated again. The State of Emergency was only canceled May 2, 2020.
Special measures in restricting people movements between municipalities
were also taken for the Easter celebrations, from 9 to 13 April, 
closing all airports to civil transportation and increased control 
in the national borders. A plan to start releasing gradually the country 
from COVID-19 container measures and canceling the State of Emergency
was only started to be implemented by May 4, 2020. 

Our goal here is to develop a mathematical model for COVID-19 able
to describe well the pandemic in Portugal,
since the emergence of the first case, on March 2,
till May 4, 2020, taking into account the real/official data
made available by ``Dire\c{c}\~{a}o Geral de Sa\'{u}de'' 
(DGS, the Portuguese Health Care Authority) \cite{dgs-covid}. 

The paper is organized as follows. In Section~\ref{Sec:model},
we introduce a SAIQH (Susceptible--Asymptomatic--Infectious--Quarantined--Hospitalized) 
mathematical model with the final purpose to investigate COVID-19 dynamics in Portugal.
The model is analyzed in Section~\ref{sec:2}, where we prove
the positivity and boundedness of solutions, we compute the basic reproduction number,
and we investigate the equilibrium points of the model and their stability.
The model is then shown to describe well the transmission dynamics of COVID-19 in Portugal
in Section~\ref{sec:num:simu}, for the mentioned period of 64 days, between 
March 2 and May 4, 2020. We end with Section~\ref{sec:conc} of conclusions and discussion.
 

\section{Model Formulation}
\label{Sec:model}

With the final purpose to investigate COVID-19 dynamics in Portugal, 
we propose a SAIQH (Susceptible--Asymptomatic--Infectious--Quarantined--Hospitalized) 
type model, based on a model analyzed in \cite{Kim} for the respiratory syndrome coronavirus 
transmission dynamics in South Korea. Here, we also consider the $H_{IC}$ class, of 
hospitalized individuals in intensive care, and the compartment $D$ of deaths due to COVID-19.
Precisely, the total living population under study at time $t$, denoted by $N(t)$, 
is divided into six classes at time $t \ge 0$: (i) the susceptible individuals $S(t)$; 
(ii) the infected individuals without (or with mild) symptoms $A(t)$ (the Asymptomatic); 
(iii) infected individuals $I(t)$ with visible symptoms; (iv) quarantined individuals $Q(t)$
in isolation at home; (v) hospitalized individuals $H(t)$; 
(vi) and hospitalized individuals $H_{IC}(t)$ in intensive care units. 
Consequently, we have
$$
N(t) = S(t) + A(t) + I(t) + Q(t) + H(t) + H_{IC}(t)
$$
for all time $t \ge 0$. Moreover, we also consider a seventh class, denoted by $D(t)$, 
that gives the cumulative number of deaths due to COVID-19 for all time $t \ge 0$.
We consider a constant recruitment rate $\Lambda > 0$ into the susceptible class $S$ 
and a constant natural death rate $\mu > 0$ for all time $t \geq 0$ under study.
Susceptible individuals can become infected with COVID-19 at a rate 
$$
\lambda(t) = \dfrac{\beta \big [l_A A(t) + I(t) + l_H H(t)\big ]}{N(t)},
$$ 
where $\beta>0$ is the human-to-human transmission rate per unit of time (day) 
and $l_A>0$ and $l_H>0$ quantify the relative transmissibility of asymptomatic 
individuals and hospitalized individuals, respectively. 
Note that the class $H_{IC}$ is not included in $\lambda(t)$ 
due to the fact that the percentage of health care workers that get infected 
by SARS-CoV-2 in intensive care units is very low and can be neglected.
A fraction $p\in[0,1]$ of the susceptible population is in quarantine at home, 
at rate $\phi>0$. Consequently, only a fraction $1-p\in[0,1]$ of susceptible individuals 
are assumed to be able to become infected. Since there is uncertainty about 
immunity after recovery, we assume that individuals of class $Q$ will 
become susceptible again at rate $\omega>0$. We also consider that only 
a fraction $m\in[0,1]$ of quarantined individuals moves from class $Q$ to $S$.
It means that $(m\times 100)\%$ of quarantined individuals will return to class $S$ 
at the end of $\frac{1}{\omega}$ days. These assumptions are justified by the 
state of calamity that was immediately decreed by the government of Portugal to address 
the COVID-19 outbreak, which was followed by 45 days of a more severe state of emergency. 
The imposed restrictions during the state of calamity and state of emergency were fully 
respected by the Portuguese population. 

After $\frac{1}{\upsilon}$ days of infection, only a fraction $q\in[0,1]$ 
of infected individuals without (or with mild) symptoms will have severe symptoms 
(see \cite{WHO:corona}). Thus, $(q\times100)\%$ of individuals of compartment $A$ 
moves to $I$, at rate $\upsilon>0$. A fraction $f_1\in[0,1]$ of infected individuals 
with severe symptoms is treated at home and the other fraction $(1-f_1)\in[0,1]$ 
is hospitalized, both at rate $\delta_1$. 

Our model consider three scenarios for hospitalized individuals:
\begin{enumerate}[i)]
\item a fraction $f_2\in[0,1]$ of individuals in class $H$ can 
evolve to a state of severe health status, thus needing 
an invasive intervention, such as artificial respiration,
and, consequently, move to intensive care units, at rate $\delta_2>0$;

\item a fraction $f_3\in[0,1]$ of individuals in class $H$ die due to COVID-19, 
the disease-related death rate associated with hospitalized individuals being $\alpha_1>0$;

\item a fraction $(1-f_2-f_3)\in[0,1]$ of individuals in class $H$ gets better and, 
consequently, return to home in quarantine/isolation, at rate $\delta_2$.
\end{enumerate}
For hospitalized individuals in intensive care units, our model considers two possibilities:
\begin{enumerate}[i)]
\item a fraction $(1-\kappa)\in[0,1]$ of individuals in class $H_{IC}$ 
gets better and moves to the class $H$, at rate $\eta>0$;

\item a fraction $\kappa\in[0,1]$ of individuals in class $H_{IC}$ dies 
due to COVID-19, the disease-related death rate associated with hospitalized 
individuals in intensive care being $\alpha_2>0$.
\end{enumerate}
Note that all individuals of classes $S$, $A$, $I$, $Q$, $H$ and $H_{IC}$ 
are subject to other death reasons, at a natural death rate $\mu$.

The previous assumptions are translated into the following
mathematical model:
\begin{align}
\label{modelo-covid19-pt}
\begin{cases}
\dot{S}(t) = \Lambda + \omega mQ(t) 
- \big[ \lambda(t)(1-p) + \phi p  + \mu \big]S(t),\\[0.15cm]
\dot{A}(t) = \lambda(t)(1-p)S(t) - (q\upsilon +\mu) A(t),\\[0.15cm]
\dot{I}(t) = q\upsilon A(t) - (\delta_1 + \mu) I(t),\\[0.15cm]
\dot{Q}(t) = \phi pS(t) + \delta_1f_1I(t) + \delta_2(1-f_2-f_3)H(t) 
- (\omega m + \mu)Q(t),\\[0.15cm]
\dot{H}(t) = \delta_1(1-f_1)I(t) + \eta(1-\kappa)H_{IC}(t) \\
\qquad\qquad - \big[\delta_2(1-f_2-f_3) 
+ \delta_2f_2 + \alpha_1f_3 + \mu\big]H(t),\\[0.15cm]
\dot{H}_{IC}(t) = \delta_2f_2H(t) - \big[\eta(1-\kappa) 
+ \alpha_2\kappa + \mu\big]H_{IC}(t),\\[0.15cm]
\dot{D}(t) = \alpha_1f_3H(t) + \alpha_2\kappa H_{IC}(t),
\end{cases}
\end{align}
which is presented in a schematic way in Figure~\ref{diagrama-covid}.
All parameter and initial conditions of our mathematical 
model~\eqref{modelo-covid19-pt} are described 
in Table~\ref{Tab_parameter_description}.
\begin{table}[!htb]
\centering
\caption[]{Description of the parameters and initial conditions 
of mathematical model \eqref{modelo-covid19-pt}.}
\label{Tab_parameter_description}
\begin{tabular}[center]{|c|l|} \hline
\textbf{Parameter} & \textbf{Description} \\ \hline \hline
\small $\Lambda$ & \small Recruitment rate \\
\small $\mu$ & \small Natural death rate \\
\small $\beta$ & \small Human-to-human transmission rate \\
\small $l_A$ & \small Relative transmissibility of individuals in class $A$ \\
\small $l_H$ & \small Relative transmissibility of individuals in class $H$ \\
\small $\phi$ & \small Rate associated with movement from $S$ to $Q$ \\
\small $\upsilon$ & \small Rate associated with movement from $A$ to $I$ \\
\small $\delta_1$ & \small Rate associated with movement from $I$ to $Q$/$H$ \\
\small $\delta_2$ & \small Rate associated with movement from from $H$ to $Q$/$H_{IC}$ \\
\small $\eta$ & \small Rate associated with movement from from $H_{IC}$ to $H$ \\
\small $\omega$ & \small Rate associated with movement from from $Q$ to $S$ \\
\small $\alpha_1$ & \small Disease-related death rate of class $H$ \\
\small $\alpha_2$ & \small Disease-related death rate of class $H_{IC}$ \\
\small $p$ & \small Fraction of susceptible individuals putted in quarantine \\
\small $q$ & \small Fraction of infected individuals with severe symptoms \\
\small $f_1$ & \small Fraction of infected individuals with severe symptoms in quarantine \\
\small $f_2$ & \small Fraction of hospitalized individuals transferred to $H_{IC}$ \\
\small $f_3$ & \small Fraction of hospitalized individuals who dies from COVID-19 \\
\small $\kappa$ & \small Fraction of hospitalized individuals in intensive care units\\
\ & \small who die from COVID-19 \\
\small $m$ & \small Fraction of individuals who moves from $Q$ to $S$ \\
\small $S(0) = S_0$ & \small Individuals in class $S$ at $t=0$ \\
\small $A(0) = A_0$ & \small Individuals in class $A$ at $t=0$ \\
\small $I(0) = I_0$ & \small Individuals in class $I$ at $t=0$ \\
\small $Q(0) = Q_0$ & \small Individuals in class $Q$ at $t=0$ \\
\small $H(0) = H_0$ & \small Individuals in class $H$ at $t=0$ \\
\small $H_{IC}(0) = H_{IC_0}$ & \small Individuals in class $H_{IC}$ at $t=0$ \\
\small $D(0) = D_0$ & \small Individuals in class $D$ at $t=0$ \\ \hline
\end{tabular}
\end{table}
\begin{figure}[ht!]
\centering
\includegraphics[scale=0.65]{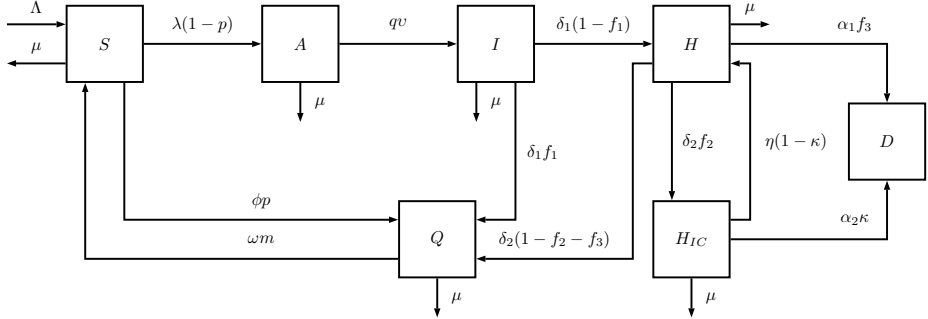}
\caption{Diagram of mathematical model~\eqref{modelo-covid19-pt}.}
\label{diagrama-covid}
\end{figure}


\section{Model Analysis}
\label{sec:2}

System \eqref{modelo-covid19-pt} is equivalent to 
\begin{equation}
\label{modelo-covid19-pt-2}
\begin{cases}
\dot{S}(t) = \Lambda + \omega mQ(t) - \big[ \lambda(t)(1-p) 
+ \phi p  + \mu \big]S(t),\\[0.15cm]
\dot{A}(t) = \lambda(t)(1-p)S(t) - (q\upsilon +\mu) A(t),\\[0.15cm]
\dot{I}(t) = q\upsilon A(t) - (\delta_1 + \mu) I(t),\\[0.15cm]
\dot{Q}(t) = \phi pS(t) + \delta_1f_1I(t) 
+ \delta_2(1-f_2-f_3)H(t) - (\omega m + \mu)Q(t),\\[0.15cm]
\dot{H}(t) = \delta_1(1-f_1)I(t) + \eta(1-\kappa)H_{IC}(t)\\ 
\qquad \qquad - \big[\delta_2(1-f_2-f_3) + \delta_2f_2 + \alpha_1f_3 + \mu\big]H(t),\\[0.15cm]
\dot{H}_{IC}(t) = \delta_2f_2H(t) - \big[\eta(1-\kappa) + \alpha_2\kappa + \mu\big]H_{IC}.
\end{cases}
\end{equation}
Note that at each instant of time $t\geq 0$ 
the cumulative number of deaths 
$D(t)$ due to COVID-19 is given by
$$
{D}(t) = D_0 + \alpha_1 f_3 \int_0^t H(\tau) d\tau
+ \alpha_2\kappa \int_0^t H_{IC}(\tau) d\tau.
$$
Throughout this section we assume $l_A = 1$. 


\subsection{Positivity and boundedness of solutions}

Since system of equations \eqref{modelo-covid19-pt-2} represents human populations, 
all parameters in the model are non-negative. It can be shown that, given non-negative 
initial values, the solutions of the system are non-negative. Precisely,
let us consider the biologically feasible region
\begin{equation*}
\Omega = \left\{ \left( S, A, I, Q, H, H_{IC}\right)
\in \left(\mathbb{R}^+_0\right)^{6} \, : \, N \leq \frac{\Lambda}{\mu} \right\}.
\end{equation*}
In what follows we prove the positive invariance of $\Omega$
(i.e., all solutions in $\Omega$ remain in $\Omega$ for all time).

\begin{lemma}
\label{lemma:posit:invariant}
The region $\Omega$ is positively invariant for the model
\eqref{modelo-covid19-pt-2} with non-negative initial conditions
in $\left(\mathbb{R}^+_0\right)^{6}$.
\end{lemma}

\begin{proof}
The rate of change of the total population, obtained by adding 
all the equations in model \eqref{modelo-covid19-pt-2}, is given by
\begin{equation*}
\frac{dN}{dt} = \Lambda - \mu N(t) - \alpha_1 f_3 H(t) - \alpha_2 k H_{IC}.
\end{equation*}
Using the standard comparison theorem \cite{Lakshmikantham,Silva:2015} 
one can easily show that
$$
N(t) \leq N(0) e^{-\mu t} + \frac{\Lambda}{\mu} \left( 1- e^{-\mu t} \right).
$$
In particular, $N(t) \leq \frac{\Lambda}{\mu}$ if $N(0) \leq \frac{\Lambda}{\mu}$.
Thus, the region $\Omega$ is positively invariant and it is sufficient to consider
the dynamics of the flow generated by \eqref{modelo-covid19-pt-2} in $\Omega$.
In this region, the model is epidemiologically and mathematically well posed 
in the sense of \cite{Hethcote} and every solution of 
\eqref{modelo-covid19-pt-2} with initial conditions 
in $\Omega$ remains in $\Omega$ for all $t > 0$.
\end{proof}


\subsection{Equilibrium points and stability analysis}

With the purpose to simplify expressions, let us introduce the
following notation:
\begin{enumerate}[i)]
\item $a_0 :=q v+ \mu$;
\item $a_1:=\delta_1+\mu$;
\item $a_2:= m \omega + \mu$;
\item $a_3 := \delta_2 (1-f_2-f_3)+\delta_2 f_2+\alpha_1 f_3+\mu$;
\item $a_4 := \delta_2 (1-f_2-f_3)$;
\item $a_5 := p \phi+\mu$;
\item $a_6 := \delta_1 (1-f_1)$;
\item $a_7 := \alpha_2 k+ \eta_k+ \mu$;
\item $\eta_k := \eta (1-k)$;
\item $\chi :=  a_3 \, a_7 -\delta_2\, \eta_k \, f_2 $. 
\end{enumerate}
Model \eqref{modelo-covid19-pt-2} has a disease free equilibrium, $\Sigma_0$, given by
\begin{equation}
\label{eq:dfe}
\Sigma_0 = \left(S_0, A_0, I_0, Q_0, H_0, H_{IC_0} \right)\\ 
= \left( \frac{ \Lambda a_2 }{ (p \phi+a_2) \mu }, 0, 0, 
\frac{p \phi \Lambda }{ (p \phi+a_2)\mu }, 0, 0\right) \, . 
\end{equation}
The local stability of $\Sigma_0$ can be established using the
next-generation operator method of \cite{Driessche}. Moreover,
following the approach in \cite{Driessche}, 
the basic reproduction number is given by 
\begin{equation}
\label{eq:R0}
R_0 = \frac{ \beta a_2 \left( 1-p \right) \big[ (l_H \, a_6 \, q v  
+ (a_1  + qv)a_3 ) a_7 - \delta_2\, \eta_k \, f_2 (qv  + a_1)\big] }{ a_0 
\, a_1 \, \chi  \left( p \phi+ a_2 \right) }  = \frac{\mathcal{N}}{\mathcal{D}}.
\end{equation}

\begin{lemma}
The disease free equilibrium $\Sigma_0$ \eqref{eq:dfe}
is locally asymptotically stable
if $R_0 < 1$ and unstable if $R_0 > 1$,
where $R_0$ is the basic reproduction number \eqref{eq:R0}.
\end{lemma}

\begin{proof}
Following Theorem~2 of \cite{Driessche}, the disease-free equilibrium (DFE),
$\Sigma_0$, is locally asymptotically stable if all the eigenvalues of the Jacobian
matrix of system \eqref{modelo-covid19-pt-2}, here denoted by $M\left(\Sigma_0\right)$,
computed at the DFE $\Sigma_0$, have negative real parts.
The Jacobian matrix $M\left( \Sigma_0 \right)$ of system \eqref{modelo-covid19-pt-2}
at disease free equilibrium $\Sigma_0$ is given by
\begin{equation} 
\label{eq:Jacob:DFE}
\left[ \begin {array}{cccccc} - a_5 & \dfrac{ a_2 \,\beta\,
\left(p -1 \right) }{\phi\,p+ a_2}&\dfrac{ a_2 \,\beta\,
\left(p -1 \right) }{\phi\,p+ a_2}&\dfrac{ l_H \, a_2 \,
\beta \, \left(p -1\right)}{\phi\,p + a_2}& m \omega & 0\\
\noalign{\medskip} 0&-\dfrac{ a_0 \,p\phi+ a_2( \beta\,p+ a_0 
- \beta)}{\phi\,p+a_2}&-\dfrac{ a_2\,\beta\, \left(p -1\right) }{\phi\,p + a_2 }
&-\dfrac{ l_H \, a_2\,\beta\, \left(p -1 \right) }{\phi\,p+ a_2} & 0 & 0 \\ 
\noalign{\medskip} 0& q v &- a_1 & 0 & 0 & 0\\ \noalign{\medskip}
0 & 0 & a_6 &- a_3 & 0& \eta_k \\ \noalign{\medskip} \phi\, p & 0& \delta_1
\, f_1 & a_4 &- a_2 & 0\\ \noalign{\medskip} 0 & 0 & 0 &\delta_2\ f_2 & 0 &- a_7 
\end {array} \right] \, .
\end{equation}
One has
\begin{equation*}
\textrm{trace}\left[ M\left(\Sigma_0\right) \right] 
= - a_5 -\frac{ a_0 \,p\phi+ a_2 \,\beta ( p - 1) 
+ a_0 \,a_2}{\phi\,p+ a_2 }
- a_1 - a_3 - a_2 - a_7   < 0 
\end{equation*}
and
\begin{equation*}
\det \left[ M\left(\Sigma_0\right) \right] 
=(\mathcal{D}-\mathcal{N})\mu  > 0
\end{equation*}
for $R_0  < 1$, where
$\mathcal{D}$ and $\mathcal{N}$ denote the denominator and numerator of $R_0$, respectively. 
We have just proved that the disease free equilibrium $\Sigma_0$ of model
\eqref{modelo-covid19-pt-2} is locally asymptotically stable if $R_0 < 1$
and unstable if $R_0 > 1$.
\end{proof}

Assume that the transmission rate is strictly positive, that is, 
\begin{equation}
\label{eq:lambdat}
\lambda^* =\dfrac{\beta (A^* + I^* + l_H H^*)(1-p) }{N^*} > 0
\end{equation} 
and $\chi > 0$. Then, model \eqref{modelo-covid19-pt-2} 
has an endemic equilibrium 
\begin{equation*}
\Sigma^* = \left(S^*, A^*, I^*, Q^*, H^*, H_{IC}^* \right), 
\end{equation*}
where
\begin{equation}
\label{eq:EE}
\begin{split}
S^*&= \dfrac{\lambda^* \, a_0 \, a_1 \, a_2 \, \chi }{ \mathcal{D}^*},\\[0.15cm]
A^*&= \dfrac{ a_1 \, a_2 \, \chi \, \Lambda \, \lambda^*}{ \mathcal{D}^* },\\[0.15cm]
I^*&= \dfrac{ \chi\, \Lambda\, a_2 \, q v \, \lambda^* }{  \mathcal{D}^* },\\[0.15cm]
Q^*&=\dfrac{ \Lambda ((\chi \delta_1 f_1 + a_4 a_6 a_7) 
q v  \lambda^*+ a_0 a_1 p \phi \chi) }{ \mathcal{D}^* },\\[0.15cm]
H^*&= \dfrac{\Lambda \, a_2\,  a_6 \, a_7 \, q v \, \lambda^*  }{ \mathcal{D}^* },\\[0.15cm]
H_{IC}^*&= \dfrac{ \delta_2\, f_2 \, \Lambda \, a_2 \, a_6 \, qv \, \lambda^* }{  \mathcal{D}^* } \, , 
\end{split}
\end{equation}
where $\mathcal{D}^* = \left(  \chi  \left( - f_1 \, m \omega \, 
q v \,\delta_1 + a_0 \, a_1 \, a_2 \right) - a_4 \, a_6 \, a_7  
m \omega  q v \right) \lambda^* + \mathcal{D} \mu$. 

Using \eqref{eq:EE} in the expression for $\lambda^*$ in \eqref{eq:lambdat} 
shows that the nonzero (endemic) equilibria of the model satisfy
\begin{equation*}
\lambda^* = \dfrac{\left(\mathcal{N}
-\mathcal{D}\right) \beta (1-p)}{\mathcal{N} + q v (a_2 a_6 (a_7 (1-l_H)
+\delta_2 f_2)+\delta_1 f_1 \chi + a_4 a_6 a_7)) \beta (1-p)} \, . 
\end{equation*}
The force of infection at the steady-state $\lambda^*$ is positive, 
only if $R_0 > 1$. We have just proved the following result.

\begin{lemma}
The model \eqref{modelo-covid19-pt-2} has a unique endemic equilibrium whenever
$R_0 > 1$.
\end{lemma}

\begin{remark}
The dynamical behavior of model \eqref{modelo-covid19-pt-2},  
when $R_0 = 1$, can be studied following the approach in 
\cite{Castillo-Chavez:MBE:2004}: see, e.g., \cite{AnaPaiao:JOTA}.
\end{remark}


\section{COVID-19 Spread in Portugal}
\label{sec:num:simu}

In this section, we show that model \eqref{modelo-covid19-pt-2} describes well
the transmission dynamics of COVID-19 in Portugal, taking into account the official
data from daily reports of ``Dire\c{c}\~{a}o Geral de Sa\'{u}de'' 
(DGS, the Portuguese Health Care Authority), available in \cite{dgs-covid}. 
We consider the COVID-19 spread in Portugal 
from the 2nd of March until 4th May 2020. Therefore,
the initial time $t = 0$ corresponds to March 2, 2020.

\begin{remark}
The data considered in this paper concerns the period from the first confirmed 
case with COVID-19 in Portugal and the period of confinement corresponding 
to the three states of emergency that were declared between March 18, 2020 
and May 2, 2020. In practice, the emergency period hold until May 4, 2020 
due to governmental measures taken during the extended weekend of May 2 and 3, 2020,
after the international workers' day. In this Section~\ref{sec:num:simu}, 
we show that the basic reproduction number \eqref{eq:R0} for the parameter values 
used (cf. Table~\ref{Tab_parameter_value}), takes the value $R_0^* = 0.95$. 
This value, close to one, indicates the sensible epidemic situation in which Portugal 
is at the end of three emergency states.
\end{remark}


\subsection{Initial conditions}
\label{subsec:initial:values}

Following the official information from \cite{dgs-covid}, on 2nd March 2020 
the two first infections by the SARS-CoV-2 virus, with symptoms, 
were confirmed in Portugal. Consequently, $I_0 = 2$.

According to \cite{WHO:corona}, only a fraction of infected individuals develops symptoms. 
Therefore, we assume that $I_0 = q \times A_0$ (see mathematical model \eqref{modelo-covid19-pt}). 
Moreover, following \cite{RTP}, we consider the value $q=0.15$, and we obtain 
$$ 
A_0 = \frac{I_0}{0.15} \simeq 13.
$$
Following \cite{dgs-covid}, we consider 
$$
H_0 = H_{IC_0} = D_0 = 0.
$$
As until 2nd March 2020 quarantine had not been advised/imposed in Portugal, 
we take $Q_0=0$.

Some of the information used in this paper is taken from 
the Portuguese database \text{PORDATA} of 2018, being 2018 the year 
with more recent available information \cite{pordata}. 
Based on \cite{pordata}, the total population in 2018 was equal 
to $10\ 283\ 800$, therefore we assume $N(0) = N_0 = 10\ 283\ 800$.

It follows that 
$$
S_0 = N_0 - A_0 - I_0 - Q_0 - H_0 - H_{IC_0} = 10\ 283\ 785.
$$


\subsection{Parameter values}
\label{subsec:par:values}

Following the information of \cite{pordata}, there were, in 2018, 
$87\ 020$ newborns in Portugal. Moreover, we assume 26\ 050 
foreign entrances in Portugal. We can conclude that there were, 
on average, $\frac{87\ 020 + 26\ 050}{365}$ newborns per day in 2018.
In agreement, we assume that 
$$
\Lambda = \frac{87\ 020 + 26\ 050}{365}.
$$

Furthermore, there were $113\ 051$ deaths in 2018, in Portugal.
Then, there were, on average, $\frac{113\ 051}{365}$ deaths 
per day in 2018 and, consequently, we take
$$
\mu = \frac{113\ 051}{365\times N_0}.
$$
On 12th March 2020, Portuguese Government suspended classes 
from 16th March 2020 (see \cite{gov-pt}). Actually, the Portuguese 
were advised to stay at home, avoiding social contacts, since 14th March 2020, 
inclusive, restricting to the maximum their exits from home.
As the two first infected cases in Portugal happened on 2nd March 2010, 
we consider that the quarantine was adopted in education after twelve days after 
the beginning of the pandemic in Portugal, that is, $\phi = \frac{1}{12}$. 
Using the Portuguese database \text{PORDATA} of 2018 \cite{pordata}, 
because there is no more recent official information,  we assume that there are:
\begin{enumerate}[i)]
\item 240\ 231 children in nurseries (see \cite{pordata-infantarios});
\item 987\ 704 students in basic education (see \cite{pordata});
\item 401\ 050 students in high education (see \cite{pordata});
\item 372\ 753 students in universities (see \cite{pordata}).
\end{enumerate}
So, 2 001 738 students were in quarantine at home, since the end of 13th March 2020.
Moreover, in 2018, there were 2\ 228\ 750 individuals with age greater or equal 
to 65 years (see \cite{pordata}). Population with these ages were also advised 
to be at home, in quarantine. We also consider that unemployed individuals 
(7\% of the total population) have decided to be in quarantine after 13th March 2020. 
As the Portuguese population in 2018 was equal to $10\ 283\ 800$, we assume that 
there are 719 866 unemployed individuals in quarantine (see \cite{pordata}).
Concluding, we assume that after twelve days of the beginning 
(12 days after after March 2, 2020), 
$$
2\ 001\ 738 + 2\ 228\ 750 + 719\ 866 = 4\ 950\ 354
$$ 
individuals began their quarantine. With students at home, some parents 
decided to work from home. Actually, two million of Portuguese were 
in this situation (see \cite{jornal_negocios}). So, we assume that the 
fraction of susceptible people in quarantine after twelve days of the beginning is
$$
p = \dfrac{4\ 950\ 354 + 2\ 000\ 000}{N_0} \simeq 0.68.
$$
Through DGS, we can consider that $q=0.15$ (see \cite{RTP}).
According to World Health Organization (WHO), $\frac{1}{\upsilon}\in [1,14]$, 
but $\frac{1}{\upsilon}=5$ on average (see \cite{WHO:corona}).

The fraction of hospitalized individuals $H$ that are transferred to the class 
$H_{IC}$, given by the parameter $f_2$, is computed taking the average value 
of the percentages of individuals in intensive care $H_{IC}$ with respect 
to the ones that are hospitalized $H$, which is equal to approximately $f_2 = 0.21$.  

The fractions of hospitalized individuals, $H$, and hospitalized individuals
in intensive care units, $H_{IC}$, who die from COVID-19, $f_3$ and $k$, 
respectively, are assumed to take the same values, because there is no available 
information about these values separately, and they are computed taking into account 
the last 42 days, according to the data in \cite{dgs-covid}, 
taking approximately the value $f_3 = k = 0.03$. 

The assumed value for the transfer rate from quarantine $Q$ to susceptible $S$, 
$\omega =  1/31$ ($day^{-1}$), makes sense in the context of the 45 days of
state of emergency and home containment obligation imposed by the Portuguese authorities.

\begin{table}[!htb]
\centering
\caption[]{Parameter values and initial conditions for model \eqref{modelo-covid19-pt}.}
\label{Tab_parameter_value}
\begin{tabular}[center]{|c|c|c|} \hline
\textbf{Parameter} & \textbf{Value} & \textbf{Reference}\\ \hline \hline
\small $\Lambda$ & \small $(87\ 020 + 26\ 050)/365$ (person day$^{-1}$)  & \small \cite{pordata}\\
\small $\mu$ & \small $113\ 051/(365\times N_0)$ (day$^{-1}$)  & \small \cite{pordata}\\
\small $\beta$ & \small 1.93 (day$^{-1}$)  & \small Assumed \\
\small $l_A$ & \small 1 (dimensionless) & \small Assumed \\
\small $l_H$ & \small 0.1 (dimensionless) & \small Assumed \\
\small $\phi$ & \small 1/12 (day$^{-1}$) & \small \cite{gov-pt}\\
\small $\upsilon$ & \small 1/5 (day$^{-1}$) & \small \cite{WHO:corona}\\
\small $\delta_1$ & \small 1/3 (day$^{-1}$) & \small Assumed \\
\small $\delta_2$ & \small 1/3 (day$^{-1}$) & \small Assumed \\
\small $\eta$ & \small 1/7 (day$^{-1}$) & \small Assumed \\
\small $\omega$ & \small 1/31 (day$^{-1}$) & \small Assumed \\
\small $\alpha_1$ & \small 1/7 (day$^{-1}$) & \small Assumed \\
\small $\alpha_2$ & \small 1/15 (day$^{-1}$)  & \small Assumed \\
\small $p$ & \small $0.68$ & \small \cite{jornal_negocios,pordata,pordata-infantarios} \\
\small $q$ & \small 0.15 & \small \cite{RTP}\\
\small $f_1$ & \small 0.96 & \small \cite{dgs-covid} \\
\small $f_2$ & \small 0.21  & \small \cite{dgs-covid} \\
\small $f_3$ & \small 0.03 & \small \cite{dgs-covid} \\
\small $\kappa$ & \small 0.03 & \small Assumed\\
\small $m$ & \small 0.075 & \small Assumed\\
\small $S_0$ & \small 10\ 283\ 785 (person) & \small \cite{dgs-covid,pordata,RTP,WHO:corona} \\
\small $A_0$ & \small 13 (person) & \small \cite{dgs-covid,RTP,WHO:corona} \\
\small $I_0$ & \small 2 (person) & \small \cite{dgs-covid} \\
\small $Q_0$ & \small 0 (person) & \small Assumed \\
\small $H_0$ & \small 0 (person) & \small \cite{dgs-covid} \\
\small $H_{IC_0}$ & \small 0 (person) & \small \cite{dgs-covid} \\
\small $D_0$ & \small 0 (person) & \small \cite{dgs-covid} \\ \hline
\end{tabular}
\end{table}


\subsection{Numerical simulations}
\label{subsec:num:sim}

Now we fit the real data from the daily Portuguese reports from \cite{dgs-covid}, 
of the values of infected $I$ and hospitalized $H$ individuals. We provide numerical 
simulations with respect to classes $I$ and $H$, using all the values 
of Table~\ref{Tab_parameter_value}.

In Figure~\ref{fig:I}, we find the predicted number of infected individuals $I$ 
in Portugal for $t\in[0,64]$, through the mathematical model \eqref{modelo-covid19-pt}, 
versus real data from 2nd March 2020 to 4th may 2020, available in \cite{dgs-covid}.
It is important to note that $I(t)$ does not represent the cumulative number 
of confirmed cases at each day $t$, but the number of infected individuals 
with symptoms at each day $t$.
\begin{figure}[ht!]
\centering
\includegraphics[scale=0.45]{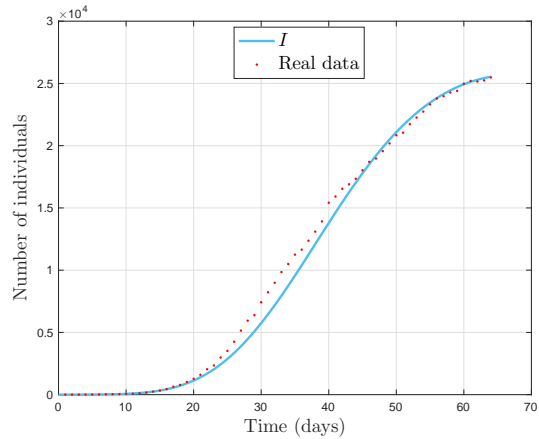}
\caption{Number of COVID-19 infected individuals $I(t)$, 
$t\in[0,64]$ days, predicted by model~\eqref{modelo-covid19-pt}, 
assuming all values of Table~\ref{Tab_parameter_value} -- solid green line;
real data of Portugal, from 2nd March 2020 to 4th May 2020, 
from \cite{dgs-covid} -- red points.}
\label{fig:I}
\end{figure}

In Figure~\ref{fig:H}, we find the predicted number of hospitalized individuals 
in Portugal for $t\in[0,64]$, through the mathematical model \eqref{modelo-covid19-pt}, 
versus real data from 2nd March 2020 to 4th May 2020, available in \cite{dgs-covid}.
However, the plotted curve $H$, predicted by the proposed mathematical 
model \eqref{modelo-covid19-pt} as given in Figure~\ref{fig:H},
has a shift of 13 days with respect to reported values 
of hospitalized individuals in Portugal.
\begin{figure}[ht!]
\centering
\includegraphics[scale=0.45]{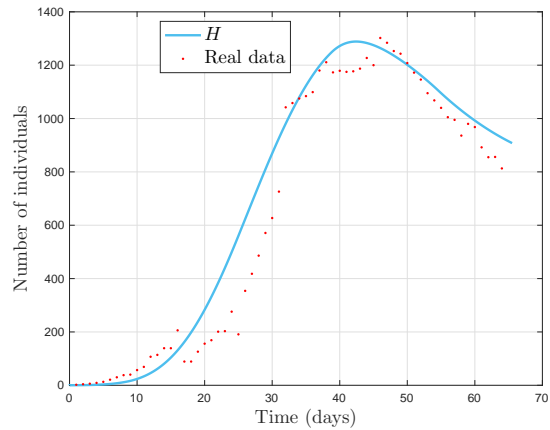}
\caption{Number of hospitalized individuals $H(t)$, $t\in[0,64]$ days, 
predicted by model~\eqref{modelo-covid19-pt}, assuming all values 
of Table~\ref{Tab_parameter_value} -- solid green line;
real data of Portugal, from 2nd March 2020 to 4th May 2020, 
from \cite{dgs-covid} -- red points.}
\label{fig:H}
\end{figure}

The basic reproduction number $R_0$, given by \eqref{eq:R0}, 
takes the value $R_0^* \simeq 0.95$ for the parameter values 
of Table~\ref{Tab_parameter_value}. The points $(\beta = 1.93, R_0^*)$, 
$(p = 0.68, R_0^*)$, $(\omega = 1/31, R_0^*)$, and $(m=0.075, R_0^*)$, 
are marked in Figures~\ref{fig:R0_beta}--\ref{fig:R0_m} with the mark $*$. 
\begin{figure}[ht!]
\centering
\includegraphics[scale=0.45]{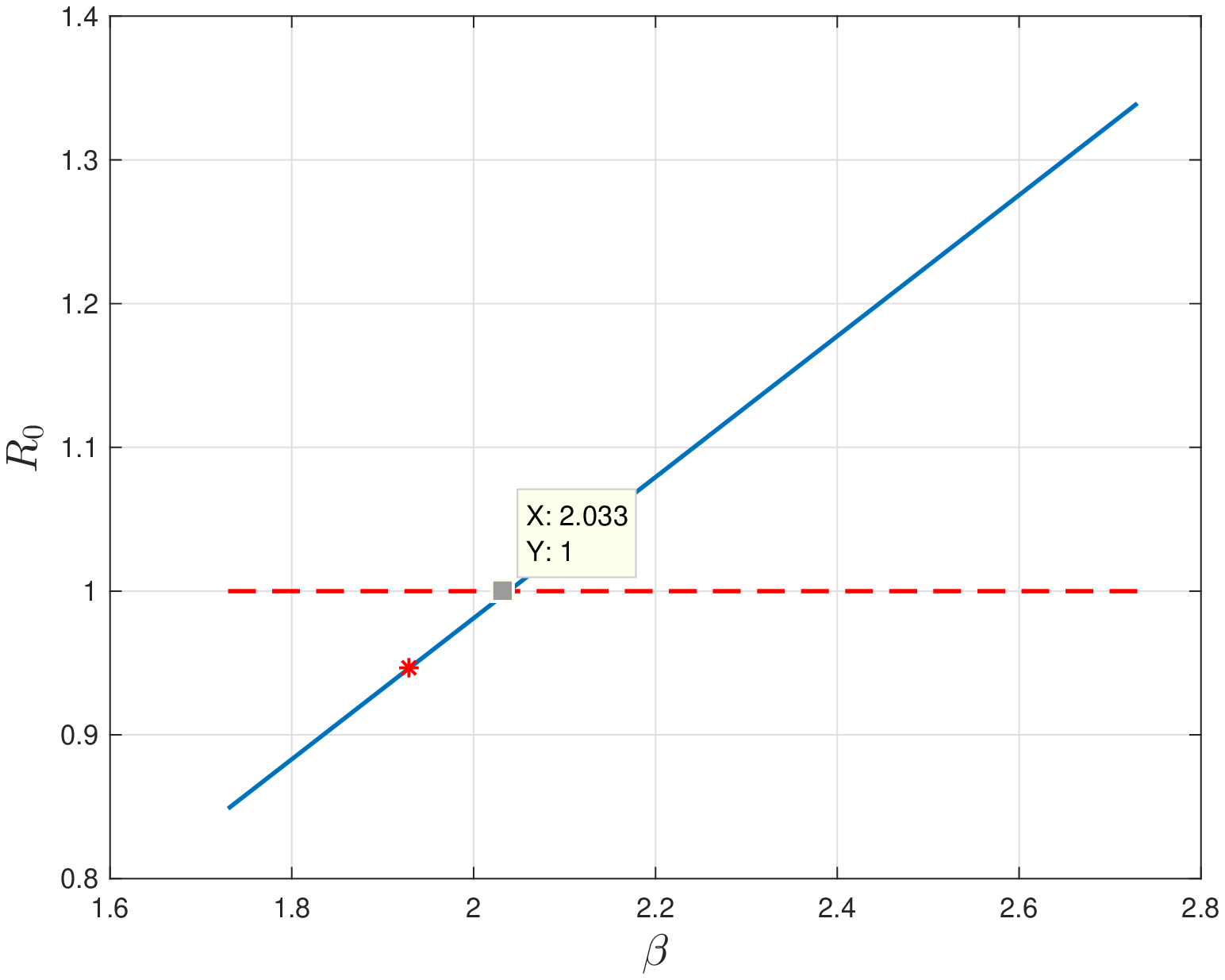}
\caption{Influence of the parameter $\beta$ in the value of $R_0(\beta)$, 
for $\beta \in [\beta^*-0.2, \beta^*+ 0.8]$, where $\beta^*$ 
and all the other parameters take the values of Table~\ref{Tab_parameter_value}.}
\label{fig:R0_beta}
\end{figure}
\begin{figure}[ht!]
\centering
\includegraphics[scale=0.45]{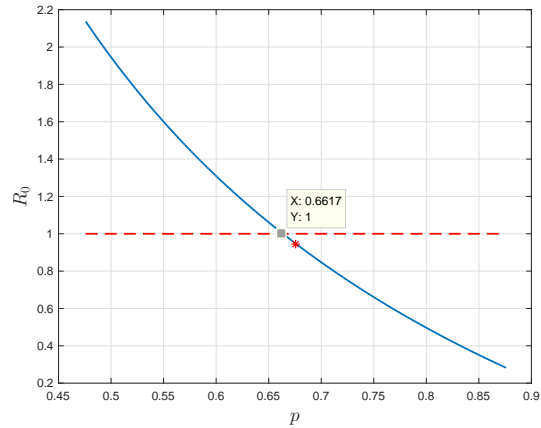}
\caption{Influence of the parameter $p$ in the value of $R_0(p)$, 
for $p \in [p^*-0.2, p^*+0.2]$, where $p^*$ and all the other parameters 
take the values of Table~\ref{Tab_parameter_value}.}
\label{fig:R0_p}
\end{figure} 
\begin{figure}[ht!]
\centering
\includegraphics[scale=0.45]{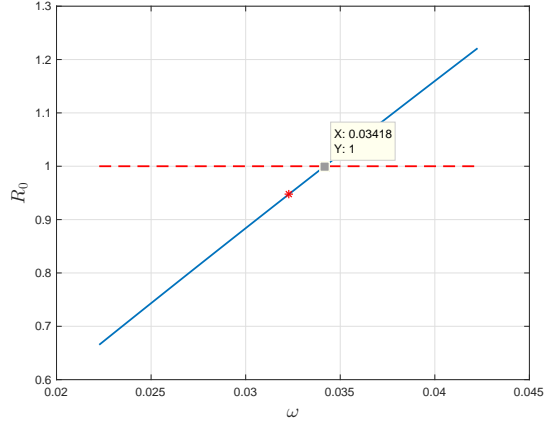}
\caption{Influence of the parameter $\omega$ in $R_0(\omega)$, 
for $\omega \in [\omega^*- 0.01, \omega^* + 0.01]$, where $\omega^*$ 
and all the other parameters take the values 
of Table~\ref{Tab_parameter_value}.}
\label{fig:R0_omega}
\end{figure} 
\begin{figure}[ht!]
\centering
\includegraphics[scale=0.45]{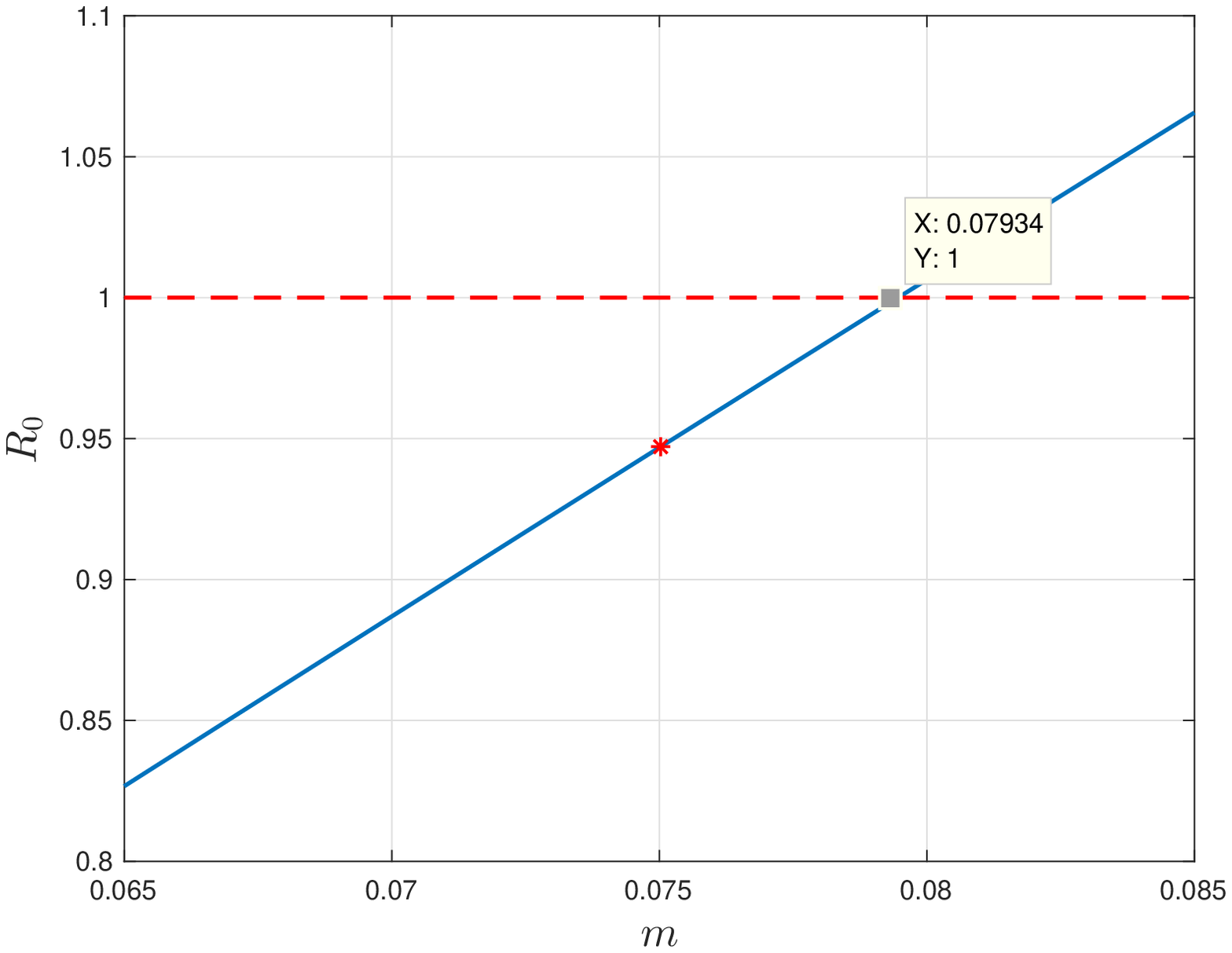}
\caption{Influence of the parameter $m$ in $R_0(m)$, for
$m \in [m^*- 0.01, m^* + 0.01]$, where $m^*$ and all the other 
parameters take the values of Table~\ref{Tab_parameter_value}.}
\label{fig:R0_m}
\end{figure} 

\begin{remark}
Our SAIQH-like compartmental epidemic model describes very well the evolution of COVID-19 in Portugal. 
Indeed, the results of the fittings are quite good: the $L^2$ error of the fit associated with 
infected individuals $I$ (see Figure~\ref{fig:I}) and hospitalized individuals $H$ 
(see Figure~\ref{fig:H}), both with respect to the initial population $N_0$, 
are equal to $9.2152e-04$ and $1.6136e-04$, respectively. Moreover, the parameter 
values we are considering, correspond to a basic reproduction number 
that is in agreement with the value given by the Portuguese authorities at the end of 
the three emergency states (cf. page 8 of the report \cite{gov-pt-2}).
\end{remark}

We can see that at the end of the 45 days period of state of emergency, 
Portugal is in a possible turning point situation where the number of 
new infected individuals can continue to decrease or, in the opposite, 
if the number of contacts between infected and susceptible individuals increases, 
the Portuguese epidemiological situation can converge to an endemic equilibrium. 
However, if social distance and all preventive measures are taken, 
the spread of corona virus can remain in values for which the Portuguese Health system 
can answer in an efficient way, has it happened until the date of 4th May, 2020.   


\section{Conclusions and Discussion}
\label{sec:conc}

SIR (Susceptible--Infectious--Recovered) 
type models provide powerful tools to describe and control 
infection disease dynamics \cite{MR3054565,MR3508846,Silva:2015}. 
Recently, several mathematical models have been developed 
to study the COVID-19 pandemic. These include 
\cite{MR4093642}, where the case of Wuhan is addressed;
\cite{MR4097809}, where the COVID-19 epidemic in Morocco is analyzed;
\cite{MR4099213}, which investigates the situation in Costa Rica;
as well as several other investigations studying the reality of many other countries: 
Canada \cite{MR4099358}, Italy \cite{Nature:It}, Pakistan \cite{MR4104871}, etc. 
Other studies analyze and compare the realities of more than one country: see, e.g., 
\cite{MR4099358}, for the cases of China, Italy and France; 
or \cite{MR4101941}, for the realities of South Korea, Italy, and Brazil. 
Typically, the proposed mathematical models fits well
the number of confirmed active infected cases with COVID-19.
Here we propose a new compartmental model that not only 
allows to fit well the number of confirmed active infected COVID-19 cases 
but also fits simultaneously the number of individuals that need to be hospitalized. 
This is far from being trivial and one should remark
that the standard SIR, SEIR, and SEIRQ models, as well as many other extensions
found in the literature, do not allow to do this,
while our model does. Such novelty is of primordial importance, since 
it is well known that one crucial point in COVID-19 pandemic 
is to ensure that the health systems are never overloaded.

We proposed a new compartmental model with the goal
to describe the COVID-19 pandemic in Portugal,
since its emergence, on 2nd March 2020, till 4th May 4, 2020, 
when the Portuguese authorities started releasing gradually the country 
from the severe COVID-19 confinement measures and canceling the very restrictive
State of Emergency measures. 

Figure~\ref{fig:I} shows a strictly increasing function in all the considered time interval, 
that is, the pandemic continues to propagate, since the peak of infected individuals
has not yet been reached at $t = 64$. Figure~\ref{fig:H} shows a strictly increasing function
in an interval $[0, t_{\max})$ and strictly decreasing in the interval of time $(t_{\max},64]$.
We conclude that the maximum number of hospitalized individuals has already been reached
and despite the fact that the pandemic continues to spread, hospitalization in Portugal is controlled.
From Figures~\ref{fig:R0_beta}, \ref{fig:R0_p}, \ref{fig:R0_omega} and \ref{fig:R0_m}, 
we observe that the famous $R_0$ depends linearly on $\beta$, $\omega$ and $m$ parameters, 
but the same is no longer the case with the $p$ parameter. Furthermore, we can see that 
as the values of $\beta$, $\omega$ and $m$ increase, the value of $R_0$ increases as well. 
Regarding the parameter $p$, the opposite happens: as the value of $p$ increases, 
the value of the basic reproduction number $R_0$ decreases.
Such numerical conclusions make sense considering the meaning of all these parameters 
(see Table~\ref{Tab_parameter_description}) and are in agreement with the analytical results 
obtained in Section~\ref{sec:2}.

Concluding, our results show that at the end of the three emergency states
declared in Portugal, the country begins 4th May 2020 a possible turning point situation 
where the number of new infected individuals can continue to decrease or, in the opposite, 
if the number of contacts between infected and susceptible individuals increases, 
the Portuguese epidemiological situation can converge to an endemic equilibrium. 
However, if social distance and all preventive measures continue to be taken, 
the spread of corona virus can remain in values for which the Portuguese 
Health system can answer in an efficient way, has it happened until the date of 4th May, 2020.  
Because the fraction of the population that was infected
is very small and nothing is known yet about the herd immunity of the population,
there is always the danger that with the gradual return of Portuguese population to the susceptible 
class, from June and July 2020 on, a very large and rapid increase in disease transmission may result. 
It means that the possibility of a second wave of COVID-19 in Portugal is not ruled out.

From the theoretical/mathematical point of view, it remains open 
the study of the local and global stability of the endemic equilibrium.
This seems quite difficult. From one side, for the local stability, the computation 
of the eigenvalues of the Jacobian matrix evaluated at the endemic equilibrium is cumbersome. 
For the global stability, finding a suitable Lyapunov function is also a difficult task. 
These are interesting open problems.

 
\section*{Acknowledgments} 

This research was supported by the Portuguese Foundation for Science and Technology (FCT) 
within ``Project Nr.~147 -- Controlo \'Otimo e Modela\c{c}\~ao Matem\'atica da Pandemia \text{COVID-19}: 
contributos para uma estrat\'egia sist\'emica de interven\c{c}\~ao em sa\'ude na comunidade'', 
in the scope of the ``RESEARCH 4 COVID-19'' call financed by FCT,
and by project UIDB/04106/2020 (CIDMA). Silva is also supported by national funds (OE), 
through FCT, I.P., in the scope of the framework contract foreseen in the numbers 
4, 5 and 6 of the article 23, of the Decree-Law 57/2016,
of August 29, changed by Law 57/2017, of July 19.
The authors are grateful to four anonymous reviewers
for many suggestions and invaluable comments,
which helped them to improve their manuscript.



\end{document}